\documentclass[superscriptaddress,twocolumn,nofootinbib]{revtex4}
\usepackage{amssymb}
\usepackage{amsfonts}
\usepackage{amsmath}
\usepackage{amsthm}
\usepackage{graphicx}
\usepackage{dcolumn}
\usepackage{bm}
\usepackage[english]{babel}
\usepackage{enumerate}
\usepackage[T1]{fontenc}

\newtheorem{proposicion}{Proposition}

\begin{document}

\title{An analysis of Born-Infeld determinantal gravity in Weitzenb\"{o}ck spacetime}

\author{Franco Fiorini}
\email{francof@cab.cnea.gov.ar.}
\affiliation{Instituto Balseiro, Universidad Nacional de Cuyo, R8402AGP Bariloche, Argentina.}

\author{Nicolas Vattuone}
\email{nicolas.vattuone@ib.edu.ar}
\affiliation{Instituto Balseiro, Universidad Nacional de Cuyo, R8402AGP Bariloche, Argentina.}


\begin{abstract}
The Born-Infeld theory of the gravitational field formulated in Weitzenb\"{o}ck spacetime is studied in detail. The action, constructed quadratically upon the torsion two-form, reduces to Einstein gravity in the low field limit where the Born-Infeld constant $\lambda$ goes to infinity, and it is described by second order field equations for the vielbein field in $D$ spacetime dimensions. The equations of motion are derived, and a number of properties coming from them are discussed. In particular, we show that under fairly general circumstances, the equations of motion are those of Einstein's General Relativity plus an energy-momentum tensor of purely geometrical character. This tensor is obtained solely from the parallelization defining the spacetime structure, which is encoded in a set of $D$ smooth, everywhere non-null, globally defined 1-forms $e^{a}$. Spherical symmetry is studied as an example, and we comment on the emergence of the Schwarzschild geometry within this framework. Potential (regular) extensions of it are envisioned.
\end{abstract}
\maketitle

\section{Introduction: Born-Infeld gravity}

Born-Infeld (BI) gravitational actions are adapted versions for the gravitational field of the (once quite innovative) BI electrodynamics, whose undeniable curative virtues concerning the singularity problem of the point-like charge electromagnetic field, were recognized since its very conception during the thirties \cite{BI1},\cite{BI2}. In a previous article \cite{Nos1}, we have introduced a novel approach to BI gravity, the so called BI determinantal gravity in Weitzenb\"{o}ck spacetime (see also \cite{DAS} and \cite{Mariam} for further results). This theory exhibits a number of regular solutions in cosmological contexts \cite{Nos2},\cite{Nos3}, where strong curvature singularities as the Big Bang are removed by a natural (i.e., geometrical) mechanism. In particular, the existence of regular black hole solutions in pure vacuum constitutes a promising feature of this theory, a property not shared with other BI-like candidates as, for instance, \emph{Eddington-Born-Infeld} gravity \cite{Max},\cite{Max3}. On the other hand, in the geometrical structure provided by the Weitzenb\"{o}ck space, the equations of motion for the vielbein field $e^{a}$ are of second order due to the fact that the gravitational action is constructed from the torsion tensor $T^{a}=de^{a}$. This remarkable aspect is certainly one that other BI-like approaches to the gravitational field would wish to possess \cite{deser}-\cite{Cagri6}. Alternative lines of research concerning BI actions for the gravitational field are being the object of considerable attention nowadays \cite{fder1}-\cite{fder4}.

In the present paper we will derive the full field equations in $D$ spacetime dimensions, coming from the determinantal structure studied previously just in highly symmetric (mostly cosmological) manifolds. The so obtained dynamical equations enable us to discuss certain important and general issues arising from them in regard to the behavior of the theory in the strong field limit, where it actually shows a more rich dynamical response than its low energy counterpart given by General Relativity (GR).  In order to proceed, we will briefly introduce the absolute parallelism (teleparallel) description of GR, as well as the deformed BI determinantal gravity based on it. For details on the former and its relationship with the gauge approach to gravity, see \cite{Hehl1}-\cite{Hehl3}, while a thorough description of the latter can be found in \cite{Nos1} and \cite{Nos2}. In what follows we shall adopt the signature $(-,+,+,...)$, and, as usual, Latin indexes $a:0,1,...$ refer to tangent-space objects, while Greek $\mu:0,1,...$ to spacetime components. Throughout the paper, the symbol $|\,|$ alludes to the absolute value of the determinant, and symmetric and skew-symmetric tensor components are expressed by $(\,)$ and $[\,]$ respectively.

In the absolute parallelism formulation of GR, the Einstein equations with cosmological constant $\Lambda$ are
\begin{equation}
|e|^{-1} \partial_\mu ( |e| S_a^{\mu \nu})-S_\rho^{~ \mu \nu} T^{\rho}_{ ~ \mu a} + \frac{1}{4} e^\nu_a (T-2\Lambda) = 4\pi G \mathbf{T}_a^{~ \nu}
\label{eq: ETRG},
\end{equation}
where $\mathbf{T}_a^{~ \nu}$ is the energy-momentum tensor of the matter fields, $T_{\ \ \mu \nu }^{\rho}=e_{a}^{\rho}\,(\partial _{\mu }e_{\nu }^{a}-\partial _{\nu}e_{\mu }^{a})$ are the components of the torsion 2-form which arises as a consequence of the non symmetric character of Weitzenb\"{o}ck connection $\Gamma _{\mu \nu }^{\lambda}=e_{a}^{\lambda }\,\partial _{\nu }e_{\mu }^{a}$, and $e_{\mu }^{a}$ are the components of the vielbein field (with inverse $e_{a }^{\mu}$), related with the components of the metric tensor according to $g_{\mu\nu}=\eta _{ab}\,e^{a}_{\mu}e^{b}_{\nu}$. The tensor $S_{\rho }^{\ \ \mu \nu }$ appearing in (\ref{eq: ETRG}) is constructed solely in terms of the torsion, and reads
\begin{equation}
S_{\rho }^{\ \ \mu \nu }=\frac{1}{4}\,(T_{\ \ \ \rho }^{\nu \mu }+T_{\rho }^{\ \ \mu \nu
}-T_{\ \ \ \rho }^{\mu \nu
})+\frac{1}{2}(\delta _{\rho }^{\mu }\,T_{\ \ \ \theta }^{\theta
\nu }-\delta _{\rho }^{\nu }\,T_{\ \ \ \theta }^{\theta \mu }).
\notag  \label{tensorS}
\end{equation}
The important spacetime scalar $T$ in (\ref{eq: ETRG}), is the so called Weitzenb\"{o}ck invariant
\begin{equation}\label{weitinvariant}
T=S_{\rho }^{\ \ \mu \nu }T^{\rho}_{\ \ \mu \nu}.
\end{equation}
The equations of motion (\ref{eq: ETRG}) can be derived from the absolute parallelism (teleparallel) action in $D$ dimensions

\begin{equation}\label{acciontele}
I=\frac{1}{16 \pi G} \int d^{D}x\,|e|\, (T-2\Lambda)+ \int d^{D}x\,|e|\, L_{\mathcal{M}},
\end{equation}
where, just as in eqs. (\ref{eq: ETRG}), $|e|\doteq |g|^{1/2}$, and $L_{\mathcal{M}}$ refers to the matter Lagrangian density. The fact that (\ref{acciontele}) is actually equivalent to the Hilbert-Einstein action, can be seen from the identity $R=-T +\ 2\,|e|^{-1}\,(|e|T^{\mu\, \, \rho}_{\ \,\mu})_{,\,\rho}$, where $R$ is the scalar curvature coming from the Levi Civita connection.

The BI gravitational action in $D$-dimensional Weitzenb\"{o}ck spacetime represents a high energy deformation of (\ref{acciontele}), and is given by
\begin{equation}
I_{\mathbf{BIG}}=\frac{\lambda}{16 \pi G} \int d^{D}x\Big[\sqrt{|g_{\mu
\nu}+2\lambda ^{-1}F_{\mu \nu }|}-\Delta\sqrt{|g_{\mu \nu }|}\Big],
\label{acciondetelectro}
\end{equation}
where
\begin{equation}
F_{\mu \nu }=\alpha \,F^{(1)}_{\mu \nu}+\beta \,F^{(2)}_{\mu \nu}+\gamma \,F^{(3)}_{\mu \nu},  \label{tensorF}
\end{equation}%
and each component $F^{(i)}_{\mu \nu}$ is defined according to
\begin{equation}
F^{(1)}_{\mu \nu}=S_{\mu }^{\;\;\lambda \rho }T_{\nu \lambda \rho}\,,\,\,\,F^{(2)}_{\mu \nu}=S_{\lambda \mu }^{\;\;\;\rho }T_{\,\,\,\,\,\nu \rho }^{\lambda
}\,,\,\,\,F^{(3)}_{\mu \nu}=\,g_{\mu \nu }\,T\,. \label{tensorFies}
\end{equation}%
Note that the metric $g_{\mu\nu}=\eta _{ab}\,e^{a}_{\mu}e^{b}_{\nu}$ is actually a subsidiary field in the action (\ref{acciondetelectro}), being the actual dynamical field, the vielbein $e^{a}$.

The dimensionless constants $\alpha ,\beta ,\gamma$ in (\ref{tensorF}), span a 2-parametric family of smooth deformations of GR. This is so because $\alpha+\beta +D\, \gamma =1$ in order to get $Tr(F_{\mu\nu})=T$. This last condition constitutes the low field limit of the theory, where the BI parameter $\lambda$ goes to infinity. In this case, the action (\ref{acciondetelectro}) recovers the teleparallel action (\ref{acciontele}) with cosmological constant $2\Lambda=\lambda(1-\Delta)$, and then, the usual general relativistic description of the gravitational field given by Einstein field equations in its absolute parallelism form\footnote{In the following analysis, we shall fix $\Delta=1$}. In regard of this, note that the left hand side of (\ref{eq: ETRG}) is just one half of the Einstein tensor with cosmological constant term $\textbf{G}^{\nu}_{a}(\Lambda)/2$, which can be written in terms of $F^{(2)}_{\mu \nu}$ defined in (\ref{tensorFies}) as
\begin{equation}
 \frac{1}{2}\textbf{G}^{\nu}_{a}(\Lambda)=|e|^{-1}\partial_{\mu}(|e|S_{a}^{ ~ \mu\nu})-F^{(2)\nu}_{ ~ ~ ~ ~ a} + \frac{1}{4}(T - 2\Lambda) e^{\nu}_{a}.
 \label{eq:ETRGbajaF}
 \end{equation}

As was discussed several times in the literature (see, e.g., \cite{Nos1},\cite{Nos2}), the action (\ref{acciondetelectro}) is not local Lorentz invariant. This is so because under the action of an element $\Delta^{a'}_{b}$ of the Lorentz group $SO(1,D-1)$, the torsion 2-form transforms as $T^{a}\rightarrow T^{a'}=\Delta^{a'}_{b}T^{b}-e^{b}\wedge d\Delta^{a'}_{b}$, and then $F_{\mu\nu}$ is invariant just under an on shell \emph{remnant} subgroup of $SO(1,D-1)$. This restricted invariance was discovered very recently in the context of $f(T)$ gravity \cite{Nos4}, and it is expected that analogous results might emerge in BI gravity by virtue of the similarities arising in the underlying geometrical structure, even though the associated Lorentz subgroup corresponding to a given spacetime can be very different in the two theories. Nevertheless, the Lorentz symmetry is protected in the low energy limit corresponding to GR, because the Weitzenb\"{o}ck invariant transforms under $SO(1,D-1)$ as $T\rightarrow T'=T+$ surface term, leaving invariant, the dynamics coming from (\ref{acciontele}).

The appearance of equivalence classes of preferred frames $[e^{a}]$ in the strong field regime, defined by frames connected through Lorentz transformations belonging to the remnant group of a given spacetime in BI gravity, constitutes in our opinion an interesting feature which may be telling us something very important about the structure of the gravitational field in such a regime. The full characterization of the remnant group of BI gravity along the lines followed in \cite{Nos4} for the $f(T)$ case, remains as an open problem. For details regarding the emergence of preferred frames in theories with absolute parallelism, the reader is invited to consult Ref. \cite{Nos5}.

\section{The field equations}
It is our aim now to derive the equations of motion coming from the action (\ref{acciondetelectro}). For this purpose, it results useful to consider the expression for the determinant $\vert U\vert$ of a given matrix $U$ in terms of the cofactors, namely (no summation in $\mu$)
\begin{equation}
\vert U\vert=  \sum\limits_{\nu} U_{\mu\nu} \vert U(\mu|\nu)\vert,
\end{equation}
where $\vert U(\mu|\nu)\vert$ is the determinant of the matrix $U(\mu|\nu)$ obtained from $U$ by removing the row $\mu$ and the column $\nu$, multiplied by $(-1)^{\mu+\nu}$. If this determinant is non zero, the components of the inverse matrix $U^{-1}$ can be written in terms of the cofactors as
\begin{equation}
(U^{-1})^{\nu\mu} = \vert U(\mu|\nu)\vert/\vert U\vert.
\label{eq:inversacofactor}
\end{equation}
It is clear then that $\vert U(\mu|\nu)\vert $ is independent of the component $U_{\mu\nu}$, because it was computed removing the entire row $\mu$. Then, $ \partial \vert U(\mu|\nu)\vert/\partial U_{\mu\nu}=0$, and we have $\partial \vert U\vert/\partial U_{\mu\nu}=\vert U(\mu|\nu)\vert$. In this way, by means of equation (\ref{eq:inversacofactor}), we get
\begin{equation}
\partial \vert U\vert/\partial U_{\mu\nu} = (U^{-1})^{\nu\mu}\,\vert U\vert,
\label{eq:derivadadet}
\end{equation}
which will be the starting point of the construction. The key component of the action is the tensor
\begin{equation}
U_{ \mu \nu} = g_{\mu \nu} + 2\lambda^{-1} F_{\mu \nu},
\end{equation}
where $F\equiv F_{\mu \nu}$ is defined in (\ref{tensorF}). Let us proceed now to calculate the Euler-Lagrange equations for the Lagrangian
\begin{equation}\label{lagresta}
\mathcal{L} =\sqrt{|U|}-\sqrt{|g|},
\end{equation}
taking into account that $\mathcal{L}$ must be varied with respect to the vielbein components $e^a_{\lambda}$, i.e,
\begin{equation}\label{Eulagresta}
\delta \mathcal{L} = \left[ \frac{\partial \mathcal{L}}{\partial e^{a}_{\lambda} } - \partial_{\gamma} \left( \frac{\partial \mathcal{L}}{\partial (\partial_\gamma e^{a}_{\lambda})} \right) \right] \delta e^a_{\lambda}.
\end{equation}
Now, the variation of the first term in (\ref{lagresta}) involves
\begin{equation}
2\delta \sqrt{ \vert U \vert } = \vert U \vert^{-\frac{1}{2}}  \delta |U| =  \vert U \vert^{-\frac{1}{2}} \partial \vert U \vert / \partial U_{\mu \nu}\, \delta U_{\mu \nu}.
\end{equation}
Using equation (\ref{eq:derivadadet}), we obtain
\begin{equation}
2\delta \sqrt{ \vert U \vert } =  \vert U \vert^{\frac{1}{2}} (U^{-1})^{\nu\mu} \delta U_{\mu \nu}.
\end{equation}
In addition, the variation of $U_{\mu \nu}$ with respect to $e^a_{\lambda}$ is given by
\begin{equation}
\delta U_{\mu \nu} =  \left[ \frac{\partial U_{\mu \nu}}{\partial e^{a}_{\lambda} } - \partial_{\gamma} \left( \frac{\partial U_{\mu \nu}}{\partial (\partial_\gamma e^{a}_{\lambda})} \right) \right] \delta e^a_{\lambda}.
\end{equation}
In the first term of this last expression we have
\begin{equation}
\frac{\partial U_{\mu \nu} }{\partial e^{a}_{\lambda}}  = \frac{\partial g_{\mu \nu}}{\partial e^{a}_{\lambda}} + \frac{2}{\lambda}  \frac{\partial F_{\mu \nu}}{\partial e^{a}_{\lambda}}, \,\,\,\,\,\, \frac{\partial g_{\mu \nu}}{\partial e^{a}_{\lambda}} = \delta^{\lambda}_{ ( \mu}e_{a\, \nu )}.
\end{equation}
In this way, the term $\partial \mathcal{L}/\partial e^{a}_{ \lambda}$ in (\ref{Eulagresta}) becomes
\begin{equation}
\frac{\partial \mathcal{L}}{ \partial e^{a}_{ \lambda} }  = \frac{|U|^{\frac{1}{2}}}{2} (U^{-1})^{\nu\mu}\left[\delta^{\lambda}_{ ( \mu}e_{a\, \nu )} + \frac{2}{\lambda}\frac{ \partial F_{ \mu \nu} }{ \partial e^{a}_{ \lambda}} \right] -   e^{\lambda}_{a} |e|,
\label{eq:Eulag1}
\end{equation}
where we have used that $\partial|e|/\partial e^{a}_{\lambda}=e^{\lambda}_{ a} |e|$.

In order to compute the second term of (\ref{Eulagresta}) we should observe that $\partial U_{\mu \nu}/\partial (\partial_\gamma e^{a}_{\lambda}) =\partial F_{\mu \nu}/\partial (\partial_\gamma e^{a}_{\lambda})$, because the metric do not depends on derivatives of the vielbein field. Then, on similar grounds we obtain
\begin{equation}
\partial \mathcal{L} /\partial (\partial_\gamma e^{a}_{ \lambda}) =  \lambda^{-1}|U|^{\frac{1}{2}} (U^{-1})^{\nu \mu}\,  \partial F_{\mu \nu}/ \partial(\partial_\gamma e^a_{ \lambda}),
\label{eq:Eulagder}
\end{equation}
so the equations of motion read\footnote{Here we have supposed that the energy-momentum tensor of spinless matter couples to the metric $g=e^{a}e^{b}\eta_{ab}$ in the usual way.}
\begin{eqnarray}
 \frac{|U|^{\frac{1}{2}}}{2} (U^{-1})^{\nu\mu}\left[\delta^{\lambda}_{ ( \mu}e_{a\, \nu )} + \frac{2}{\lambda}\frac{ \partial F_{ \mu \nu} }{ \partial e^{a}_{ \lambda}} \right] -   e^{\lambda}_{ a} |e|-  \nonumber\\
 \lambda^{-1}\partial_\gamma\Big(|U|^{\frac{1}{2}} (U^{-1})^{\nu \mu}\frac{\partial F_{\mu \nu} }{\partial(\partial_\gamma e^a_{ \lambda})}\Big) = 16 \pi G  |e|  \mathbf{T}^{\lambda}_{~  a}  .
\label{varfin}
\end{eqnarray}
However, a complete characterization of the equations of motion involves the specification of $\partial F_{ \mu \nu}/\partial e^{a}_{ ~ \lambda}$ and $\partial F_{\mu \nu}/\partial(\partial_\gamma e^a_{ ~ \lambda})$. For this purpose, let us factorize $F_{\mu \nu}$ as following. We define the tensor $D^{b ~ ~ ~ ~ \sigma \rho}_{~ \alpha \beta c}$ in such a way that $S^b_{ ~ \alpha \beta} = D^{b ~ ~ ~ ~ \sigma \rho}_{~ \alpha \beta c} T^{c}_{~ \sigma \rho}$. Its explicit form then reads
\begin{equation}\label{tensorD}
D^{b ~ ~ ~ ~ \sigma \rho}_{~ \alpha \beta c} = \frac{1}{4} \left( \delta^b_c \delta^\sigma_\alpha \delta^{\rho}_{\beta} - e^{\rho b}  e_{c [\alpha} \delta^{\sigma}_{\beta]} \right) - \frac{1}{2} e^b_{[ \alpha} \delta^{\rho}_{\beta ]} e^{\sigma}_{c}.
\end{equation}
The advantage of this factorization is that the components of the 2-form $S^b_{ ~ \alpha \beta}$ are written as the product of a tensor which depends only on the vielbein, times a tensor which depends solely on the derivatives of it. In a similar way, we can also define $Q^{\lambda ~ b}_{ ~ a ~ \alpha \beta}$, which verifies
\begin{equation}
Q^{\lambda ~ b}_{ ~ a ~ \alpha \beta} := \frac{ \partial S^{b}_{ ~ \alpha \beta} }{ \partial e^{a}_{\lambda} } =  \frac{ \partial \left( D^{b ~ ~ ~ ~ \sigma \rho}_{~ \alpha \beta c}  \right) }{ \partial e^{a}_{\lambda} }  T^{c}_{\sigma \rho},
\end{equation}
and then is explicitly expressed as
\begin{eqnarray}\label{tensorQ}
 Q^{\lambda ~ b}_{ ~ a ~ \alpha \beta} &=& \frac{1}{4} \left( T_{ [ \alpha, \beta] a} e^ { \lambda b} - \delta^{\lambda}_{[\alpha } T_{a \beta] }^{ ~ ~ ~ b} \right)- \nonumber \\&&\frac{1}{2} \left( \delta^b_a \delta^ { \lambda}_ {[\alpha } T^{\sigma}_{  ~ \sigma \beta]} - e^{b}_{[\alpha} T^{\lambda}_{ ~ a\beta]} \right).
\end{eqnarray}

With the tensors $D^{b ~ ~ ~ ~ \sigma \rho}_{~ \alpha \beta c}$ and $Q^{\lambda ~ b}_{ ~ a ~ \alpha \beta}$ so introduced, the rest of the calculations are just a matter of laborious work. At the end, we obtain for $\partial F_{ \mu \nu}/\partial e^{a}_{ ~ \lambda}$ and $\partial F_{\mu \nu}/\partial(\partial_\gamma e^a_{ ~ \lambda})$ the following expressions:

\begin{eqnarray}
&&\frac{\partial F_{ \mu \nu}}{\partial e^{a}_{ ~ \lambda}}=\alpha \Big(\delta^{\lambda}_{\mu} F^{^{(1)}}_{a \nu} + \delta^{\lambda}_{\nu} F^{^{(1)}}_{ \mu a} + Q^{\lambda}_{~ a \mu \alpha \beta} T_{\nu}^{~ \alpha \beta} - \notag\\
&&2 S_{ \mu \rho (a} T_{\nu}^{~ \rho \lambda)}\Big) +\beta \Big( Q^{\lambda}_{ ~ a \beta \mu \alpha} T^{\beta ~ \alpha}_{ ~ \nu} - S_{\beta \mu}^{ ~ ~ (\lambda} T^{ \beta}_{ ~ \nu a)}\Big) +\notag\\
&&\gamma \Big( \delta^{\lambda}_{(\mu} e_{a \nu)} T - 4 g_{\mu \nu} F^{^{(2)} \lambda}_{ ~ ~ ~ ~ a} \Big),
\label{deftet}
\end{eqnarray}

\begin{eqnarray}
&&\frac{\partial F_{\mu \nu}}{\partial(\partial_\gamma e^a_{ ~ \lambda})}=\alpha \Big(2 S_{\mu}^{~ \lambda \gamma} e_{ \nu a}+ D_{\mu ~ ~ a}^{~ \sigma \rho ~ [\lambda \gamma] } T_{\nu \sigma \rho}\Big)+\notag\\
&&\beta\Big(D_{ \beta \mu \alpha a}^{~ ~ ~ ~ [\lambda \gamma] } T^{\beta ~ \alpha}_{~ \nu}  + S_{a \mu}^{ ~ ~ [\gamma } \delta^{ \lambda]}_{\nu}\Big) +  4\gamma\, g_{\mu \nu} S_{a}^{ ~ \lambda \gamma}.
\label{defdertet}
\end{eqnarray}
We conclude this section by mentioning that, as in other general-relativistic field theories \cite{Hehl4}, automatic conservation of energy-momentum is not guaranteed in BI gravity through the equations of motion (\ref{varfin}). This can be seen easily noting that, for the case $\alpha=\beta=0$ (i.e. $\gamma=D^{-1}$) in eq. (\ref{tensorF}), BI determinantal gravity reduces to a particular $f(T)$ theory, where energy-momentum conservation is not automatic (see, for instance, Ref. \cite{Christian} regarding this point).

\section{Some implications of the field equations}

One of the first questions we intended to answer, is under what circumstances it is possible to reobtain GR's solutions from the Born-Infeld scheme under consideration. Of course, the limit $\lambda\rightarrow\infty$ assures the low energy regime provided by GR, and any solution of BI gravity will be a solution of Einstein's theory in this limit. Yet, it results crucial to develop certain criteria in order to know when a GR solution will solve the full BI determinantal field equations instead. Bearing this in mind, we can actually prove the following

\begin{proposicion}
Let $g$ be a solution of Einstein's equations with cosmological constant $\Lambda$ and energy-momentum tensor $ \mathbf{T}^{\mu \nu} $. Let be $e^a$ a vielbein such that:

\begin{itemize}
\item[(i)] It generates the metric, that is: $g = e^a e^b \eta_{ab}$
\item[(ii)] It satisfies that $ F_{ \mu \nu} \left( e^a \right)  =  C g_{\mu \nu}$ being $C\neq-\lambda/2$, a constant.
\end{itemize}

If we choose $C$ such that

\begin{equation}\label{constcosmo}
2 \widetilde{\Lambda} = \widetilde{C} (D-2) -1 + (1 + 2\widetilde{C})^{1-\frac{D}{2}},
\end{equation}

where $\widetilde{\Lambda}=\Lambda/\lambda$ and $\widetilde{C}=C/\lambda$, then $e^a$ is a solution of the determinantal equations with energy-momentum tensor $ \mathbf{T}^{' \mu \nu} = \left( 1 + 2 \widetilde{C}\right)^{\frac{D}{2}-1}   \mathbf{T}^{\mu \nu} $.
\label{prop3}
\end{proposicion}

\emph{Remark 1:} Note that if a given vielbein $e^a$ satisfies the hypothesis of the proposition with $C=0$, then $e^a$ is a solution of both theories with $\Lambda=0$ and the same matter content encoded in $\mathbf{T}^{\mu \nu}$.

\emph{Remark 2:} In a way, the physical interpretation of this proposition can be the following: a vielbein satisfying the condition $(ii)$ with $C\neq0$, \emph{generates} a cosmological constant. The value of $\widetilde{\Lambda}$ so obtained is a consequence of the specific choice of $e^a$, via the constant $C$ appearing in $(ii)$. In other words, a given GR solution can always be obtained in the context of BI gravity if an appropriated frame $e^a$ can be chosen in which $(ii)$ holds, with $C$ satisfying (\ref{constcosmo}).

\begin{proof}
If $F_{\mu \nu} = C g_{ \mu \nu}$, by definition we have $U_{\mu \nu} = \left( 1 + 2\widetilde{C} \right) g_{\mu \nu} $, which implies that $ \vert U \vert^{\frac{1}{2}} = \left( 1 + 2\widetilde{C} \right)^{ \frac{D}{2} }|e| $ and $(U^{-1})^{\nu\mu} =  g^{\nu \mu}/ \left( 1 + 2\widetilde{C} \right)$ ($(U^{-1})^{\nu\mu}$ is well defined because $C\neq-\lambda/2$). Under these conditions, (\ref{eq:Eulag1}) becomes
\begin{eqnarray}
\frac{ \partial \mathcal{ L} }{ \partial e^{a}_{\lambda} } &=&  |e|   \left[ \left( 1 + 2\widetilde{C} \right)^{\frac{D}{2} -1} - 1 \right]  e^{\lambda}_a   +\nonumber\\
&&\lambda^{-1}  |e|\left( 1 + 2\widetilde{C} \right)^{\frac{D}{2}-1} g^{\nu \mu}  \frac{ \partial F_{ \mu \nu}  }{ \partial e^{a}_{\lambda} }.
\end{eqnarray}
On the other hand, after contracting eq. (\ref{deftet}) with $g^{ \nu \mu}$, we obtain in the present circumstances,
\begin{equation}
g^{ \nu \mu} \frac{ \partial F_{ \mu \nu}  }{ \partial e^{a}_{\lambda} } = 2C e^{\lambda}_{a} - 4 F^{^{(2)} \lambda}_{ ~ ~ ~ ~ a},
\end{equation}
and then
\begin{equation}\label{par1}
\frac{ \partial  \mathcal{L} }{ \partial e^{a}_{\lambda} } = - 4 |e|  \left(  1 + 2 \widetilde{C} \right)^{ \frac{D}{2} -1 } \left[ F^{^{(2)} \lambda}_{ ~ ~ ~ ~ a} - \frac{1}{4} \left( T - 2\Lambda \right) e^{\lambda}_{a} \right],
\end{equation}
where we have used (\ref{constcosmo}), and that the trace of $F_{\mu\nu}$ is, due to $(ii)$, $T= Tr(F)= C D$.

In a similar manner, after a considerable amount of work we can write the equation (\ref{eq:Eulagder}) as
 \begin{eqnarray}\label{par2}
\frac{ \partial \mathcal{L} }{ \partial (\partial_{ \gamma} e_{\lambda}^a) } &=& \left(1 + 2\widetilde{C} \right)^{ \frac{D}{2} -1} |e| \, g^{\nu \mu} \frac{ \partial F_{\mu \nu}}{\partial (\partial_{\gamma} e_{\lambda}^a)}\nonumber\\
  &=& 4 \left(1 + 2\widetilde{C} \right)^{ \frac{D}{2} -1 } |e| S_{a}^{ ~ \lambda \gamma},
\end{eqnarray}
where we have used in the last step the expression for $g^{\nu \mu} \partial F_{\mu \nu}/\partial (\partial_{ \gamma} e_{\lambda}^a)$ provided by the contraction of $g^{\nu \mu}$ with eq. (\ref{defdertet}). With (\ref{par1}) and (\ref{par2}) at hand, we can finally write down the variation of the gravitational Lagrangian involved in the field equation (\ref{varfin}), according to
\begin{eqnarray}
&&\frac{\delta\mathcal{L}}{\delta e_{\lambda}^a}=-4 \left(  1 + 2\widetilde{C} \right)^{ \frac{D}{2} -1 }\, \times\, \nonumber \\
&&\left[ |e|  \left( F^{^{(2)} \lambda}_{ ~ ~ ~ ~ a} - \frac{1}{4} \left( T - 2\Lambda \right) e^{\lambda}_{a} \right) + \partial_{\gamma} \left( |e| S_{a}^{ ~ \lambda \gamma}\right)\right]
 \label{eq:ETRGconctecosmologica}
 \end{eqnarray}

The term in the brackets in (\ref{eq:ETRGconctecosmologica}) corresponds to the Einstein equations with cosmological constant, see eqn. (\ref{eq:ETRGbajaF}). Therefore, we can write this term as $-|e| \, \textbf{G}^{\lambda}_{~ a}(\Lambda)/2$. Due to the fact that $g_{\mu \nu}$ is a solution of Einstein equations with energy-momentum tensor $\mathbf{T}^{ \mu \nu}$, then we have

\begin{eqnarray}
 2 \left(  1 + 2\widetilde{C} \right)^{ \frac{D}{2} -1 } |e|  \textbf{G}^{\lambda}_{ ~ a} (\Lambda)&=& 16\pi G \left( 1 + 2\widetilde{C} \right)^{ \frac{D}{2} -1 } |e| \mathbf{T}^{\lambda}_{~ a}\nonumber \\
 &=& 16\pi G  |e| \mathbf{T'}^{\lambda}_{~ a},
\end{eqnarray}
which establishes the desired result.\end{proof}

\bigskip

By virtue of $(ii)$, the vielbein in the proposition above leads to a constant Lagrangian density given by $(1 + 2\widetilde{C})^{D/2}-1$. It results interesting to figure out under what circumstances (other than $C=0$, a situation contemplated in \emph{Remark 1} above), the determinantal action becomes null. In general, this condition demands $ \vert g_{\mu \nu} + \frac{2}{\lambda} F_{\mu \nu} \vert ^{1/2} = \vert g_{\mu \nu} \vert ^{1/2}$, which will be fulfilled if $F_{\mu \nu}$ is nilpotent of grade $K$, this is, if $F^{K}=0$ and $F^{K-1}\neq0$. This property allows us to prove the following
\begin{proposicion}
Let ${e^a}$ be a vielbein field that makes $F_{\mu \nu}$ nilpotent of grade $K$. Then, the field equations (\ref{varfin}) can be written as
$ \textbf{G}^\lambda_a = 8\pi G (\mathbf{T}^{\lambda} _{ ~ a} + T^{ (\textbf{g}) \lambda}_{~ ~ ~ ~ a})$,
where $ \textbf{G}^\lambda_a $ is the Einstein tensor for the metric $g=e^{a}e^{b}\eta_{ab}$, $\mathbf{T}^{\lambda} _{ ~ a}$ is the energy-momentum tensor of the matter fields, and $ T^{ (\textbf{g}) \lambda}_{~ ~ ~ ~ a}$ can be interpreted as a geometric energy-momentum tensor coming from the Lagrangian

\begin{equation}\label{lagrangeomet}
\mathcal{L}^{(\textbf{g})} = \sum\limits_{n=1}^{K-1} \left( -\frac{2}{\lambda} \right)^n \frac{|e|}{n} Tr(F^{n}).
\end{equation}
\label{prop4}
\end{proposicion}

\begin{proof}
The nilpotent character of $F_{\mu \nu}$ assures $ \vert U \vert ^{1/2} = |e|$. On the other hand we can express formally the inverse of $U^{\nu \mu}$ as
\begin{equation}
(U^{-1})^{\nu\mu} = g^{\nu \mu}+ \sum\limits_{n=1}^{\infty} \left( - \frac{2}{\lambda} \right)^n (F^{n})^{\nu \mu},
\label{serieformal}
\end{equation}
being $(F^{n})^{\nu \mu}=F^{\nu}_{\sigma1}F^{\sigma1}_{\sigma2}...F^{\sigma_{n-1}\mu}$ which, in general, will converge only if $F_{\mu\nu}$ is sufficiently small compared with $\lambda$. Irrespective of these formal issues, we have now
\begin{equation}
(U^{-1})^{\nu\mu} = g^{\mu \nu}+\sum\limits_{n=1}^{K-1} \left(-\frac{2}{\lambda} \right)^n (F^{n})^{\nu \mu},
\label{serieformalnilp}
\end{equation}
which is well defined given the $K$-nilpotency of $F_{\mu\nu}$. This expression allow us to write (\ref{eq:Eulag1}) and (\ref{eq:Eulagder}) in the form

\begin{eqnarray}
&&\lambda|e|^{-1}\frac{ \partial \mathcal{ L} }{ \partial e^{a}_{\lambda} } =  g^{\nu \mu} \frac{ \partial F_{\mu \nu} }{ \partial e^{a}_{\lambda} } - F^{(\lambda}_{ ~ ~ a)}   +\nonumber\\
 &&\sum\limits_{n=1}^{K-1} \left( - \frac{2}{\lambda} \right)^n \left( - (F^{n+1})^{(\lambda}_{  ~ ~ ~ a)}   +   (F^{n})^{ ~ \nu \mu} \frac{ \partial F_{\mu \nu} }{ \partial e^{a}_{\lambda} }\right),
\label{eq:prop4eq1}
\end{eqnarray}
\begin{eqnarray}
&&\lambda|e|^{-1}\frac{ \partial \mathcal{L} }{ \partial (\partial_{ \gamma} e_{\lambda}^a) } = g^{\nu \mu}  \frac{ \partial F_{\mu \nu} }{ \partial (\partial_{ \gamma} e_{\lambda}^a)} +  \nonumber\\
&&\sum\limits_{n=1}^{K-1} \left( - \frac{2}{\lambda} \right)^n (F^{n})^{\nu \mu}\frac{ \partial F_{\mu \nu} }{ \partial (\partial_{ \gamma} e_{\lambda}^a)  }.
\label{eq:prop4eq2}
\end{eqnarray}
With these two terms at hand we can express the variation $\delta\mathcal{L}/\delta e_{\lambda}^a$ as

\begin{eqnarray}
&&\lambda\frac{\delta\mathcal{L}}{\delta e_{\lambda}^a}=|e|\left( g^{\nu \mu}\frac{\delta F_{\mu\nu}}{\delta e_{\lambda}^a}-F^{(\lambda}_{ ~ ~ a)}\right)+\nonumber\\
&&|e|\sum\limits_{n=1}^{K-1} \left( - \frac{2}{\lambda} \right)^n \left[(F^{n})^{\nu\mu}\frac{\delta F_{\mu\nu}}{\delta e_{\lambda}^a}- (F^{n+1})^{(\lambda}_{  ~ ~ ~ a)}\right].
\label{vardeltatot}
\end{eqnarray}
As we can see, the terms in the parenthesis in the right hand side reduces to the Einstein equations (see (\ref{eq:ETRGconctecosmologica}) and the paragraph below it), therefore, it is possible to regroup these terms to obtain the Einstein tensor. In order to proceed, let us see how the terms in the square brackets arise from the variation of a new lagrangian $\mathcal{L}^{(\textbf{g})}$. In view of this, we analyze the variation of the following functional
\begin{equation}\label{vartrazaf}
\delta \left(Tr( F^{n})|e|\right) = Tr(F^{n}) \delta |e| + \delta \left(Tr(F^{n}) \right) |e|,
\end{equation}
the first term of it being null because nilpotent matrices are traceless. Let us calculate the variation of the trace, which we can write as
\begin{equation}
\delta \left( F^{\sigma_1}_{~ \sigma_2} F^{\sigma_2}_{~ \sigma_3} ... F^{\sigma_n}_{~ \sigma_1}    \right) = n F^{(n-1)\nu}_{ ~ ~ ~ ~ ~ ~ ~ \rho} ~ \delta( F^{\rho}_{~ \nu}).
\end{equation}
Considering $\delta( F^{\rho}_{~ \nu})$ as
\begin{equation}
\delta F^{\rho}_{~ \nu} = \delta \left( F_{\mu \nu} g^{ \mu \rho} \right)= g^{\mu \rho} \delta \left( F_{ \mu \nu} \right) + \delta g^{\mu \rho} F_{\mu \nu},
\end{equation}
and recalling that $ \partial g^{\mu \rho}/ \partial e^a_{\lambda} = - g^{\lambda (\mu} e^{\rho)}_{a}$ we finally get
\begin{equation}
 \delta \left(Tr(F^{n}) |e| \right) =  n\left[ - F^{n(\lambda}_{ ~ ~ ~ ~ ~ a)}  \delta e^a_{\lambda} + F^{(n-1) \nu \mu} \delta (F_{\mu \nu}) \right],
\end{equation}
which looks like the terms in brackets in (\ref{vardeltatot}). Inspired by this, we define the lagrangian (\ref{lagrangeomet}) and construct $T^{(\textbf{g})}$ as
\begin{equation}\label{tensemgeo}
T^{ (\textbf{g}) \lambda}_{ ~ ~ ~ ~ a} = \frac{-1}{8\pi G|e|} \left[\frac{ \partial \mathcal{ L}^{(\textbf{g})} }{ \partial e^{a}_{\lambda} } - \partial_\gamma \left( \frac{ \partial \mathcal{L}^{(\textbf{g})} }{ \partial (\partial_{ \gamma} e_{\lambda}^a) } \right) \right],
\end{equation}
for which the equations of motion coming from (\ref{vardeltatot}) result
\begin{equation}\label{formasug}
 \textbf{G}^\lambda_a =  8\pi G (\mathbf{T}^{\lambda} _{ ~ a} + T^{ (\textbf{g})\lambda}_{~ ~ ~ ~ a}).
\end{equation}\end{proof}

 Formally, a decomposition of the sort (\ref{formasug}) can always be obtained in the context of modified gravity, but it would result useless, unless an explicit expression for the energy-momentum tensor is displayed. The relevance of this proposition, thus, relies on the implications concerning the regularization properties underlying the determinantal theory. In fact, the expression (\ref{formasug}) shows that under these conditions the theory generates an effective energy-momentum tensor $T^{(\textbf{g}) \lambda}_{~ ~ ~ ~ a}$ whose nature comes from the determinantal structure itself. This contribution could lead to a violation of the necessary conditions for the validity of the singularity theorems, provided that it were able to generate a repulsive gravitational regime in the strong field limit.

\section{Spherical symmetry}

Now we proceed to further discuss some additional aspects concerning the field equations, and to obtain, by means of the propositions proved in the former section, certain important solutions of them. In order to do so, it is important to bear in mind the difficulties involved due to the fact that the gravitational action is not local Lorentz invariant, and that the field equations determine the full vielbein components, not merely those of the metric. Hereafter we shall set $D=4$.

The propositions stated before enable us to obtain spherically symmetric solutions of the equations of motion. This can be seen by asking what kind of frames reproduce the Schwarzschild geometry of GR. In order to give a definitive answer to this question let us consider the \emph{asymptotic frame}
\begin{equation}\label{asframe}
e^0 = A(\rho) dt\,\,\,\,\, e^{i} = B(\rho) dx_{i},
\end{equation}
associated to the isotropic line element
\begin{equation}
ds^2= - A(\rho)^2 dt^2 + B(\rho)^2 \delta^{ij}dx_{i}dx_{j}.
\label{eq: metricaisotropica}
\end{equation}
In the above equations, $x_{i}$ are cartesian coordinates and
\begin{equation}
A(\rho) = \frac{ 2 \rho - M}{ 2 \rho + M }, ~ ~ ~ ~ ~ ~ B(\rho)= \left( 1 + \frac{M}{2\rho} \right)^2,
\end{equation}
where $\rho$ is the isotropic radial coordinate related to the usual Schwarzschild radial coordinate $r$ according to $(r^2-2Mr)^{1/2} + r - M = 2\rho$. Even though the condition $T=0$ is just a necessary one in order to obtain the Schwarzschild geometry also in the BI framework under consideration (see Remark 1 above), it enables us to proceed in the same fashion as in Ref. \cite{Nos6} in the hope to fulfill the stronger condition $F_{\mu\nu}=0$ referred as in Proposition 1 (with $C=0$). Keeping this in mind, we can perform a $t-r$ boost to the asymptotic frame (\ref{asframe}). The so obtained $(\tilde{e}^0,\tilde{e}^{i})$ frame reads

\begin{eqnarray}
&&\tilde{e}^0=A(\rho)\gamma(\rho)dt-B(\rho)\Pi_{1}[x_{1}dx_{1}+x_{2}dx_{2}+x_{3}dx_{3}], \notag \\
&&\tilde{e}^{1}=-A(\rho)\Pi_{1}x_{1}dt+\notag \\
&&B(\rho)\Big[(1+\Pi_{2}x_{1}^2)dx_{1}+\Pi_{2}x_{1}x_{2}dx_{2}+\Pi_{2}x_{1}x_{3}dx_{3}\Big],\notag \\
&&\tilde{e}^{2}=-A(\rho)\Pi_{1}x_{2}dt+\notag\\
&&B(\rho)\Big[\Pi_{2}x_{1}x_{2}dx_{1}+(1+\Pi_{2}x_{2}^2)dx_{2}+\Pi_{2}x_{2}x_{3}dx_{3}\Big],\notag \\
&&\tilde{e}^{3}=-A(\rho)\Pi_{1}x_{3}dt+\notag \\
&&B(\rho)\Big[\Pi_{2}x_{1}x_{3}dx_{1}+\Pi_{2}x_{2}x_{3}dx_{2}+(1+\Pi_{2}x_{3}^2)dx_{3}\Big],\notag \\
\label{boosteadas}
\end{eqnarray}%

where $\Pi_{1}=\sqrt{\gamma^{2}(\rho)-1}/\rho$, $\Pi_{2}=(\gamma(\rho)-1)/\rho^2$, and the usual definitions for the Lorentz boost were adopted:
\begin{equation}
\gamma(\rho)=\Big(1-\beta^{2}(\rho)\Big)^{-\frac{1}{2}},\,\,\,\,\,
\beta(\rho)=v(\rho)/c. \label{defboost}
\end{equation}
After some standard calculations, the Weitzenb\"{o}ck invariant for the boosted frame (\ref{boosteadas}) can be obtained, resulting
\begin{equation}
T(\tilde{e}) = -\frac{64 \rho^2}{(M+2\rho)^4} \left(\frac{M^2 + 4\rho^2}{M^2 - 4 \rho^2}\,Z-\rho\, Z'+1\right),
\end{equation}
where $Z=Z(\rho)\doteq Cosh(\beta(\rho))$. It is straightforward to show that $T(\tilde{e})$ vanishes if
\begin{equation}\label{funcC}
Z(\rho)=-\frac{ M^2 + 4 \rho^2 + \chi \rho}{M^2-4\rho^2},
\end{equation}
where $\chi$ is an arbitrary integration constant which officiates as a boost generator. The divergent character of (\ref{funcC}) at the black hole horizon $\rho=M/2$, is just a consequence of the bad behavior of the isotropic chart there. Needless to say, this sort of pathologies can be easily circumvented by working in the maximal analytic extension provided by the Kruskal chart.

As mentioned before, the vanishing of $T(\tilde{e})$ does not implies the vanishing of $F_{\mu\nu}$, but the former condition significantly simplifies the form of the latter. Actually, the diagonal components of $F_{\mu\nu}$ coming from the frame (\ref{boosteadas}) with $Z(\rho)$ given by (\ref{funcC}), are (in isotropic coordinates $(t,\rho,\theta,\phi)$):

\begin{eqnarray}\label{comdiagf}
F_{tt}&=&-\frac{16 \rho^{4}(\chi+4M)(8 \alpha M+\chi\beta)}{(M+2\rho)^{8}}\nonumber\\
F_{\rho\rho}&=&-\frac{\rho\beta\,\chi^{2}+(2\alpha+\beta)(M^{2}+4\rho^{2})^{2}\chi+32M^{2} \alpha\rho}{\rho(M^{2}-4\rho^{2})^{2}}\nonumber\\
F_{\theta\theta}&=&\frac{\chi \rho\, (2 \alpha +\beta )}{2 (M+2 \rho)^2},\,\,\,\,F_{\phi\phi}=\frac{\chi \rho \,(2 \alpha +\beta ) \sin ^2(\theta)}{2(M+2 \rho)^2}.
\end{eqnarray}
Analogously, we can compute the off-diagonal components of $F_{\mu\nu}$. The only non vanishing ones are

\begin{eqnarray}\label{comnondiagf}
F_{t\rho}&=&-\frac{4 \rho^{3/2} (\chi+4 M)}{(M^{2}-4\rho^{2})(M+2\rho)^{4}}\times\\
 &&\frac{\rho\,(32\alpha M^{2}+\beta \chi^{2})+\chi\,(2\alpha+\beta)(M^{2}+4\rho^{2})}{(\chi+8\rho)^{1/2}(2M^{2}+\chi \rho)^{1/2}}\nonumber\\
F_{\rho t}&=&-\frac{4\rho^{3/2}(8M\alpha+\chi\beta)(\chi+8\rho)^{1/2}(2M^{2}+\chi \rho)^{1/2}}{(M-2\rho)(M+2\rho)^{5}}\nonumber.
\end{eqnarray}
The components of $F_{\mu\nu}$, as (\ref{comdiagf}) and (\ref{comnondiagf}) reveal, depend on $\alpha$, $\beta$, $M$, and on the boost generator $\chi$. In turn, the absence of $\gamma$ is due to the fact that $F^{(3)}_{\mu \nu}= T g_{\mu \nu}=0$, and then, it plays no role at all under the present circumstances.

If happen that $M\neq0$, further constraints in the parameter space must be taken into account in order to guarantee the vanishing of $F_{\mu\nu}$. In particular, the purely angular sector of $F_{\mu\nu}$ will vanish only if $\chi(2\alpha + \beta)=0$ (see (\ref{comdiagf})), which lead us to the two following possibilities:

\bigskip
\emph{Case 1: $\chi=0$}. In this case, the non vanishing components of $F_{\mu \nu}$ are:
\begin{eqnarray}\label{comdiagfchi0}
F_{tt}&=&-\frac{512 M^2 \rho^4 \alpha }{(M+2 \rho)^8},\,\,\,\,F_{\rho\rho}=-\frac{32 M^2 \alpha }{\left(M^2-4 \rho^2\right)^2}\nonumber\\
F_{t\rho}&=&F_{\rho t}=-\frac{128 M^{2} \rho^{2} \alpha}{(M-2 \rho)(M+2 \rho)^5}.
\end{eqnarray}
We see, then, that $\alpha=0$ in order to fulfill $F_{\mu \nu}=0$.

\bigskip
\emph{Case 2: $2\alpha+\beta=0$}. In this particular case we have
\begin{eqnarray}\label{comdiagfchiotro}
F_{tt}&=&-\frac{16\beta\rho^4  \left(\chi^2-16 M^2\right)}{(M+2 \rho)^8}\nonumber\\
F_{\rho\rho}&=&-\frac{\beta \left(\chi^2-16 M^2\right)}{(M^{2}-4 \rho^{2})^2} \nonumber\\
F_{t\rho}&=&-\frac{4\beta \rho^{5/2}(\chi^{2}-16M^{2})(\chi+4M)}{(M^{2}-4\rho^{2})(M+2\rho)^{4}(\chi+8\rho)^{1/2}(2M^{2}+\chi \rho)^{1/2}} \nonumber\\
F_{\rho t}&=&-\frac{4 \beta \rho^{3/2} (4M-\chi)(\chi+8\rho)^{1/2}(2M^{2}+\chi \rho)^{1/2}}{(M-2 \rho)(M+2 \rho)^5}.
\end{eqnarray}
Clearly, the condition in order to obtain $F_{\mu\nu}=0$ is to select the boost generator according to $\chi=4M$.

The results just obtained, together with Proposition 1, allow us to conclude that Schwarzschild spacetime is obtained in BI gravity, provided we set $\chi=\alpha=0$ in one hand, or $2\alpha+\beta=0$, $\chi-4M=0$, on the other. In the first case, the family of Born-Infeld gravitational theories having the Schwarzschild frame as a solution of the equations of motion, is characterized by $F_{\mu\nu}=(1-4\gamma)F^{(2)}_{\mu\nu}$ with $\gamma\neq1/4$, but otherwise free (remember that $F^{(3)}=0$). This means that $F^{(2)}_{\mu\nu}$ is identically null in this case. In the second, in turn, we have $F_{\mu\nu}=(1-4\gamma)(-F^{(1)}_{\mu\nu}+2F^{(2)}_{\mu\nu})$ (i.e., $F^{(1)}_{\mu\nu}=2F^{(2)}_{\mu\nu}$). These are genuinely different theories, and the parallelization behind them is also different because it corresponds to two radial boosts with different generator $\chi$.

In view of these results, it would seem natural to ask whether the Schwarzschild geometry arises only for these specific theories, or if it can be obtained also within the more general framework provided by the three pieces $F^{(i)}_{\mu\nu}$. In order to answer this important question, let us go back to the case with $\chi=0$, and to the corresponding components of $F_{\mu\nu}$ given in (\ref{comdiagfchi0}). A quick check shows that $F_{\mu\nu}$ is actually nilpotent of grade 2, i.e. $F^{\mu}_{~ \sigma} F^{ \sigma}_{ ~ \nu}=0$ (note that eq. (\ref{comdiagfchi0}) implies $Tr(F_{\mu\nu})=0$). Then, Proposition 2 applies in this case. The associated geometric energy-momentum tensor turns out to be null here because the geometry in question also solves the Einstein's field equations (see eq. (\ref{formasug})). This means that the frame (\ref{boosteadas}) with $\chi=0$, give us the vacuum Einstein field equations adapted to spherical symmetry, and then, lead us to the Schwarzschild solution for all values of $\alpha$ and $\beta$. Note that $F^{(3)}=0$ since the very beginning, because the frame (\ref{boosteadas}) with the boost (\ref{funcC}) gives $T=0$ for all $\chi$. In other words, the Schwarzschild geometry remains as a solution of the BI gravity with arbitrary parameters $\alpha$ and $\beta$ in the action, provided $\chi=0$ in (\ref{funcC}).

We have seen that Schwarzschild spacetime emerges out from BI gravity for all $\alpha$ and $\beta$ (provided $\chi=0$), and for $2\alpha+\beta=0$ (provided $\chi=4M$). This means that we have obtained two different parallelizations for the same spacetime, linked by a local Lorentz transformation which represents a radial boost generated by $\chi$. This radial boost, then, must be contained within the spherically symmetric remnant group \cite{Nos3} of the particular theory given by $2\alpha+\beta=0$, whose full characterization remains as an open problem at the present.

Of course, $M=0$ corresponds to Minkowski spacetime, even though the form of the tetrad field is quite involved because of the boost performed. As a matter of fact, the presence of the $\chi$-term in (\ref{funcC}) makes $Z(\rho)$ different from one when $M=0$. In this way, the Euclidean frame  $e^{0}=dt,\,e^{i}=dx_{i}$ obtained from (\ref{asframe}) by setting $M=0$, is being transformed by means of a radial boost with $\beta(\rho)=ArcCosh(1+\chi/4\rho)$ (see eq. (\ref{funcC})). Evidently, both frames lead to a consistent massless solution of the field equations in vacuum.

\section{Concluding comments}
We proceed now to summarize and further comment on the preceding results. The equations of motion of Born-Infeld determinantal gravity formulated in Weitzenb\"{o}ck spacetime, released in their entirety here for the first time, had revealed a number of important properties which will be the starting point for future developments. In particular, we have stated two propositions that provide crucial information regarding the parallelization process underlying the dynamics of the theory. These two results, used to characterize the emergence of the Schwarzschild spacetime within the conceptual body of the theory, will be the starting point for any further study concerning spherical symmetry in BI gravity.

At first glance, the results obtained in section IV concerning the appearance of the Schwarzschild solution, seem a bit disappointing. However, a closer examination indicates that it is neither fair nor convenient to maintain this pessimistic point of view. Operatively speaking, the first proposition studied above offers a methodological tool to reobtain GR solutions within the full determinantal theory; just take a vielbein $e^{a}$ representing a GR solution $g=e^{a}e^{b}\eta_{ab}$ (possibly with cosmological constant), and act over it with a suitable element $\Delta^{a}_{b}$ of the Lorentz group $SO(1,D-1)$ in such a way that $\tilde{e}^{a}=\Delta^{a}_{b}e^{b}$ assures the condition $F_{\mu\nu}(\tilde{e}^{a})=C g_{\mu\nu}$. The transformed frame $\tilde{e}^{a}$ is then a solution of the determinantal field equations, which for the case $C=0$, will represent the same GR spacetime given by $g$ (see Remark 1). But the frame $\tilde{e}^{a}$ will be, in general, highly non trivial (see, e.g., eq. (\ref{boosteadas}) for the Schwarzschild case).  This non triviality is representative of the lack of Lorentz invariance of BI gravity in the strong field regime, and  the fact that the equations of motion determine the full tetrad components, not just those of $g$, establishing in this way a certain spacetime parallelization (which can be non unique, as we commented in the lasts paragraphs of section IV). It is our conviction that the preferred frames so arising (connected one each other by transformations belonging to the remnant group of the spacetime in question), are carriers of crucial information regarding the dynamics of the gravitational field in the very strong field regime.

Nevertheless, and quite importantly, Proposition 1 provides no information about potential deformations of GR geometries, except for the fact that they will be inevitably represented by a determinantal tensor such that $F_{\mu\nu}\neq C g_{\mu\nu}$. This is actually what occurs in the previous, regular solutions reported in \cite{Nos1}, \cite{Nos2} and \cite{Nos3}. Then, the proposition leaves entirely open the question regarding the existence of regular, vacuum black holes within BI gravity. The full field equations (\ref{varfin}) will be of capital importance at this quest.

To some extent, Proposition 2 has longer range consequences, even though it does not include proposition 1 on formal grounds, except for the particular case $K=1$ (i.e., $F_{\mu\nu}=0$). It actually allowed us to conclude that, given the 2-nilpotency of the frame (\ref{boosteadas}) with $\chi=0$, the Schwarzschild geometry remains as a solution of BI gravity for free $\alpha$ and $\beta$, i.e., for the whole 2-parametric families of BI theories. What it makes proposition 2 so valuable, is the fact that it provides an explicit expression for the geometric energy-momentum tensor arising as a consequence of the non trivial parallel 1-form field underlaying the parallelization process. This tensor $T^{ (\textbf{g})\lambda}_{~ ~ ~ ~ a}$, defined in (\ref{tensemgeo}) and (\ref{lagrangeomet}), represents a key component in the understanding of the regularity aspects of the theory, and its properties concerning the energy conditions involved in the singularity theorems will be a matter of future research.

\textbf{\emph{Acknowledgments}}. The authors want to thanks D. Mazzitelli for the many constructive comments he made about this project. This work was supported by CONICET and Instituto Balseiro. F. F. is member of Carrera del Investigador Científico.


\begin{thebibliography}{99}

\bibitem{BI1} M. Born and L. Infeld, Proc. R. Soc. A 144 (1934) 425.
\bibitem{BI2} M. Born and L. Infeld, Proc. R. Soc. A 147 (1934) 522.
\bibitem{Nos1} R. Ferraro and F. Fiorini, Phys. Lett. B 692 (2010) 206.
\bibitem{DAS} S. Jana, Phys. Rev. D 90 (2014) 124007.
\bibitem{Mariam} M. Bouhmadi-Lopez, C-Y Chen and P. Chen, Phys. Rev. D 90 (2014) 123518.
\bibitem{Nos2} F. Fiorini, Phys. Rev. Lett. 111 (2013) 041104.
\bibitem{Nos3} F. Fiorini, Phys. Rev. D 94 (2016) 024030.
\bibitem{Max} M. Ba\~{n}ados, Phys. Rev. D 77 (2008) 123534.
\bibitem{Max3} M. Ba\~{n}ados and P. G. Ferreira, Phys. Rev. Lett. 105 (2010) 011101.
\bibitem{deser} S. Deser and G.W. Gibbons, Class. Quant. Grav. 15 (1998) L35.
\bibitem{Fein2} J. A. Feingenbaum, P.O. Freund and M. Pigli, Phys. Rev. D 57 (1998) 4738.
\bibitem{Comelli} D. Comelli and A. Dolgov, JHEP 0411 (2004) 062.
\bibitem{Cagri1} I. Gullu, T. Cagri Sisman and B. Tekin, Class. Quant. Grav. 27 (2010) 162001.
\bibitem{Cagri2} I. Gullu, T. Cagri Sisman and B. Tekin, Phys. Rev. D 81 (2010) 104018.
\bibitem{Cagri3} I. Gullu, T. Cagri Sisman and B. Tekin, Phys. Rev. D 82 (2010) 024032.
\bibitem{Cagri5} I. Gullu, T. Cagri Sisman and B. Tekin, Phys. Rev. D 91 (2015) 044007.
\bibitem{Cagri6} I. Gullu, T. Cagri Sisman and B. Tekin, Phys. Rev. D 92 (2015) 104014.
\bibitem{fder1} J. Beltran Jimenez, L. Heisenberg, G. J. Olmo and C. Ringeval JCAP 1511:046 (2016).
\bibitem{fder2} S. D. Odintsov, G. J. Olmo and D. Rubiera-Garcia, Phys. Rev. D 90 (2014) 044003.
\bibitem{fder3} A. N. Makarenko, S. D. Odintsov and G. J. Olmo, Phys. Rev. D 90 (2014) 024066.
\bibitem{fder4} A. N. Makarenko, S. D. Odintsov and G. J. Olmo, Phys. Lett. B 734 (2014) 36.
\bibitem{Hehl1} F.W. Hehl, Y. Ne'eman, J. Nitsch and P. Von der Heyde, Phys. Lett. B 78 (1978) 102.	
\bibitem{Hehl2} J. Nitsch and F.W. Hehl, Phys. Lett. B 90 (1979) 98.
\bibitem{Hehl3} F.W. Hehl, J. D. McCrea, E. W. Mielke and Y. Ne'eman, Phy. Rept. 258 (1995) 1.
\bibitem{Nos4} R. Ferraro and F. Fiorini, Phys. Rev. D 91 (2015) 064019.
\bibitem{Nos5} R. Ferraro and F. Fiorini, Phys. Lett. B 702 (2011) 75.
\bibitem{Hehl4} F. W. Hehl and J. D. McCrea, Found. Phys. 16 (1986) 267.
\bibitem{Christian} C. G. Boehmer, A. Mussa and N. Tamanini, Class. Quant. Grav. 28 (2011) 245020.
\bibitem{Nos6} R. Ferraro and F. Fiorini, Phys. Rev. D 84 (2011) 083518.



\end{thebibliography}
\end{document}